\documentclass[12pt]{article}
\usepackage{geometry}

\usepackage{cite}
\usepackage{url}
\usepackage{amssymb,amsmath,amsthm}
\usepackage{bm}
\usepackage{graphicx}
\usepackage{enumerate}
\usepackage{microtype}
\usepackage{color}
\usepackage{hyperref}
\usepackage{multirow}
\usepackage{array}
\usepackage{float}
\usepackage{subcaption}
\usepackage{booktabs}
\usepackage{xcolor}
\usepackage{colortbl}

\newtheorem{theorem}{Theorem}

\newcolumntype{C}[1]{>{\centering\let\newline\\\arraybackslash\hspace{0pt}}m{#1}}

\begin{document}

\title{Beltrami coefficient and angular distortion of discrete geometric mappings}
\author{Zhiyuan Lyu$^{1}$, Gary P. T. Choi$^{1,\ast}$\\
\\
\footnotesize{$^{1}$Department of Mathematics, The Chinese University of Hong Kong}\\
\footnotesize{$^\ast$To whom correspondence should be addressed; E-mail: ptchoi@cuhk.edu.hk}
}
\date{ }

\maketitle

\begin{abstract}
Over the past several decades, geometric mapping methods have been extensively developed and utilized for many practical problems in science and engineering. To assess the quality of geometric mappings, one common consideration is their conformality. In particular, it is well-known that conformal mappings preserve angles and hence the local geometry, which is beneficial in many applications. Therefore, many existing works have focused on the angular distortion as a measure of the conformality of mappings. More recently, quasi-conformal theory has attracted increasing attention in the development of geometric mapping methods, in which the Beltrami coefficient has also been considered as a representation of the conformal distortion. However, the precise connection between these two concepts has not been analyzed. In this work, we study the connection between the two concepts and establish a series of theoretical results. In particular, we discover a simple relationship between the norm of the Beltrami coefficient of a mapping and the absolute angular distortion of triangle elements under the mapping. We can further estimate the maximal angular distortion using a simple formula in terms of the Beltrami coefficient. We verify the developed theoretical results and estimates using numerical experiments on multiple geometric mapping methods, covering conformal mapping, quasi-conformal mapping, and area-preserving mapping algorithms, for a variety of surface meshes in biology and engineering. Altogether, by establishing the theoretical foundation for the relationship between the angular distortion and Beltrami coefficient, our work opens up new avenues for the quantification and analysis of surface mapping algorithms.
\end{abstract}

\textbf{Keywords:} Angular distortion, Beltrami coefficient, conformal mapping, quasi-conformal theory, geometric mapping

\section{Introduction}

In many science and engineering problems, a fundamental task is to handle shape data and analyze the geometric features of different objects. To achieve this, it is common to consider geometric mapping methods. In particular, geometric mappings allow one to transform a shape into another one, thereby facilitating comparisons between shapes and simplifying many computational tasks. Over the past several decades, numerous efforts have been devoted to the development of geometric mapping methods. Surface parameterization~\cite{floater2005surface,sheffer2006mesh} aims to map geometrically complex surfaces onto some simpler domains. For instance, genus-0 closed surfaces can be parameterized onto a unit sphere, while simply-connected open surfaces can be parameterized onto a planar domain such as a rectangle or a unit disk.

Note that the quality of the geometric mappings is an important consideration in many practical applications. In particular, the distortions or errors induced by the mapping methods may largely affect the subsequent processing and analysis tasks of the surfaces. Therefore, it is desirable that the mappings possess low geometric distortion and can maintain certain important properties of the given shape. By a classical result in differential geometry, it is well-known that isometric (length-preserving) mappings are generally impossible. Specifically, isometric mappings between two surfaces are possible only if the surfaces have the same Gaussian curvatures~\cite{do2016differential}. In many practical problems, one may need to map a surface onto another surface with a highly different geometry, or parameterize a non-planar surface onto a planar domain, in which the above condition is easily violated. Therefore, in most cases, the mappings will unavoidably cause some distortions in either angle, area, or both. Because of this theoretical limitation, many prior works have focused on the development of angle-preserving (conformal) mapping methods, such as harmonic energy minimization~\cite{gu2004genus,lai2014folding} and relevant linearization methods~\cite{angenent1999laplace,haker2000conformal,choi2015flash}, circle patterns~\cite{kharevych2006discrete}, spectral conformal map~\cite{springborn2008conformal}, Ricci flow~\cite{jin2008discrete,yang2009generalized}, slit map~\cite{yin2008slit}, curvature flow~\cite{kazhdan2012can,crane2013robust}, conformal energy minimization~\cite{yueh2017efficient,tan2025robust}, partial welding~\cite{choi2020parallelizable}, and Dirichlet energy minimization~\cite{liao2022convergence}. There have also been multiple works on area-preserving mapping methods, such as Lie advection~\cite{zou2011authalic}, optimal mass transport~\cite{zhao2013area}, density-equalizing map~\cite{choi2018density,choi2020area,lyu2024spherical,yao2026toroidal}, stretch energy minimization~\cite{yueh2019novel,yueh2023theoretical,sutti2024riemannian}, and authalic energy minimization~\cite{liu2024convergent,liu2026spherical}. Besides, some existing works have considered achieving a balance between the angular and area distortion~\cite{liu2008local,wang2018novel,lyu2024bijective,choi2025hemispheroidal}. In many of the above works, the given surfaces are discretized as triangle meshes. Therefore, to assess the angle-preserving property of the surface mappings, a common approach is to directly evaluate the angle difference between the same discrete triangle mesh elements under the mapping, which gives an intuitive measure of the angular distortion. 

More recently, quasi-conformal theory~\cite{lehto1973quasiconformal,gardiner2000quasiconformal,ahlfors2006lectures} has been introduced to the field of geometric mapping and parameterization for both planar shapes~\cite{weber2012computing,jones2013planar,lam2014landmark,mosleh2025data} and surfaces~\cite{saucan2008local,zeng2010supine,zeng2011registration,lui2014teichmuller,choi2015fast,choi2024fast} (see~\cite{choi2023recent} for a recent survey). In many of the above works on quasi-conformal mappings, a central concept is the Beltrami coefficient, which is a complex-valued function encoding important geometric properties of a quasi-conformal map. Also, the Beltrami coefficient has been widely utilized in many practical applications in science and engineering, such as biological and medical morphometric analysis~\cite{choi2018planar,choi2020tooth,guo2023automatic}, image restoration~\cite{lau2019restoration}, surface remeshing~\cite{lyu2024bijective,lyu2024ellipsoidal}, and isogeometric analysis~\cite{pan2018low,pan2022constructing,pan2023g1}. Note that the Beltrami coefficient can be used as a measure of the conformal distortion of the mapping. In particular,  it is well-known that a mapping is conformal if its Beltrami coefficient is 0, and a Beltrami coefficient with a smaller norm corresponds to a mapping closer to being conformal. However, the precise connection between the concepts of Beltrami coefficients and angular distortions has not been fully understood. In this work, we establish the theoretical relationship between the Beltrami coefficient and the angular distortion of discrete geometric mappings (see Fig.~\ref{fig:brain_disk}). Specifically, we elucidate how the Beltrami coefficients can be used for estimating the angular distortion. As shown in Fig.~\ref{fig:brain_disk}, the discrete, face-based angular distortion measure constructed in our work matches the Beltrami coefficient plot both qualitatively and quantitatively. Motivated by the observed similarities between the measures, we establish a series of theoretical results and bounds in both the continuous case and the discrete case. We further verify our established results by evaluating the conformality of a large variety of mapping examples in terms of the Beltrami coefficient and the angular distortion, from which we can see that the two concepts are highly related and consistent. Altogether, our study provides new insights into the quantification and evaluation of the properties of different geometric mapping methods for various practical applications in science, engineering, and medicine.

\begin{figure}[t]
    \centering
    \includegraphics[width=\linewidth]{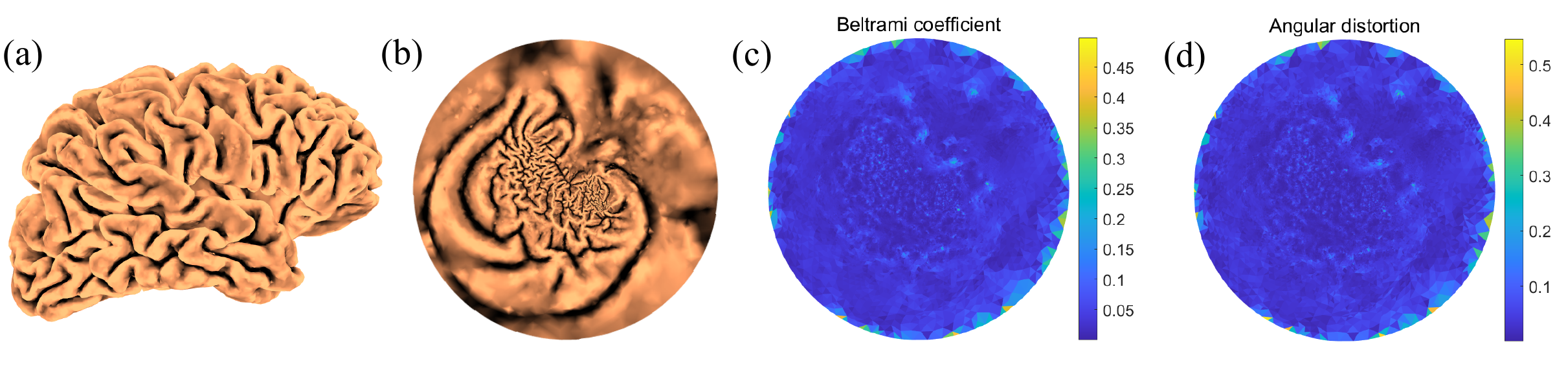}
    \caption{\textbf{The relationship between the Beltrami coefficient and angular distortion of discrete geometric mappings.} (a)--(b) Given a discrete surface such as a brain cortical surface, one can utilize surface parameterization methods to map it onto another domain. Here, the surface is parameterized onto the unit disk using the method in~\cite{choi2015fast}. (c) To assess the conformal distortion of the map, one way is to compute the Beltrami coefficient $\mu$ and evaluate the norm $|\mu|$ on each triangle element. Here, the mapping result is color-coded with the value of $|\mu|$ on each triangle element. (d) We may also directly evaluate the angular distortion by computing the angle difference for every triangle element under the mapping. Here, the mapping result is color-coded with a face-based angular distortion measure defined in our work, calculated using the average value of the absolute angle difference for each triangle. }
    \label{fig:brain_disk}
\end{figure}

The rest of this paper is organized as follows. In Section~\ref{sec:background}, we introduce some fundamental concepts in conformal mappings and quasi-conformal theory. In Section~\ref{sec:main}, we describe our main results on the theoretical connection between the Beltrami coefficient and the angular distortion of discrete geometric mappings. In Section~\ref{sec:experiment}, we present experimental results to verify the established estimates using a large variety of surface examples. We conclude this work and discuss future directions in Section~\ref{sec:conclusion}.

\section{Mathematical Background} \label{sec:background}
In this section, we review some basic concepts of conformal mapping and quasi-conformal theory. More details can be found in~\cite{lehto1973quasiconformal,gardiner2000quasiconformal,ahlfors2006lectures}.

\subsection{Conformal mapping}

Consider a map $f:\overline{\mathbb{C}} \to \overline{\mathbb{C}}$ on the extended complex plane with $f(z) = f(x,y) = u(x,y) + i v(x,y)$, where $z = x+iy$, and $u(x,y), v(x,y)$ are two functions in $x$ and $y$. We say that $f$ is a \emph{conformal} map if its derivative is non-zero everywhere and it  satisfies the Cauchy--Riemann equations:
\begin{equation}\label{eqt:cauchyriemann}
    \frac{\partial u}{\partial x} = \frac{\partial v}{\partial y}  \ \ \text{ and } \ \ \frac{\partial u}{\partial y} = -\frac{\partial v}{\partial x}.
\end{equation}
It is well-known that conformal maps preserve angles. Consequently, the local geometry is preserved under conformal maps. Intuitively, conformal maps send infinitesimal circles to infinitesimal circles. However, it is noteworthy that lengths are not necessarily preserved under conformal maps.

Also, note that if we write 
\begin{equation}\label{eqt:fzbar}
    \frac{\partial f}{\partial \overline{z}} = f_{\overline{z}} = \frac{1}{2}\left(\frac{\partial f}{\partial x} + i\frac{\partial f}{\partial y}\right),
\end{equation}
then we can rewrite Eq.~\eqref{eqt:cauchyriemann} as
\begin{equation}\label{eqt:cauchyriemann_z}
    \frac{\partial f}{\partial \overline{z}} = 0.
\end{equation}
The above equation is naturally connected to the equation for quasi-conformal maps, as we will discuss later.

Besides, the notion of conformal mapping can be naturally extended from the complex plane to surfaces. A surface $\mathcal{S}$ with a conformal structure is called a \emph{Riemann surface.} Given two Riemann surfaces $\mathcal{M}$ and $\mathcal{N}$, a map $f: \mathcal{M} \rightarrow \mathcal{N}$ is called a conformal map if satisfies
\begin{equation}
    f^{*}ds^2_{\mathcal{N}} = \lambda ds^2_{\mathcal{M}}
\end{equation}
where $\lambda$ is called the conformal factor. It is clear that $f$ preserves both angles and the infinitesimal shapes of surfaces, but omits their size or curvature. More generally, we have the following uniformization result:

\begin{theorem}(Uniformization of Riemann surfaces). Every simply connected Riemann surface $\mathcal{M}$ is conformally equivalent to exactly one of the following three domains:
\begin{enumerate}[(i)]
    \item The Riemann sphere,
    \item The complex plane,
    \item The open unit disk.
\end{enumerate} 
\end{theorem}

In other words, given any simply connected Riemann surfaces $\mathcal{M}$, we can always find a conformal map $f$ between $\mathcal{M}$ and one of the three domains.

\subsection{Quasi-conformal theory}
Quasi-conformal mapping is a generalization of conformal mapping, which allows bounded conformal distortions. Mathematically, a map $f: \mathbb{C} \rightarrow \mathbb{C}$ is called a \emph{quasi-conformal map} if it satisfies the Beltrami equation
\begin{equation}\label{eqt:Beltrami_eq}
    \frac{\partial f}{\partial \overline{z}} = \mu(z) \frac{\partial f}{\partial z}
\end{equation}
for some complex-valued function $\mu$ with $\|\mu \|_{\infty}<1$, where $\frac{\partial f}{\partial \overline{z}}$ is given by Eq.~\eqref{eqt:fzbar} and $\frac{\partial f}{\partial z}$ is defined as
\begin{equation}
    \frac{\partial f}{\partial z} = f_{z} =  \frac{1}{2}\left(\frac{\partial f}{\partial x} - i\frac{\partial f}{\partial y}\right).
\end{equation}

\begin{figure}[t]
    \centering
    \includegraphics[width=\textwidth]{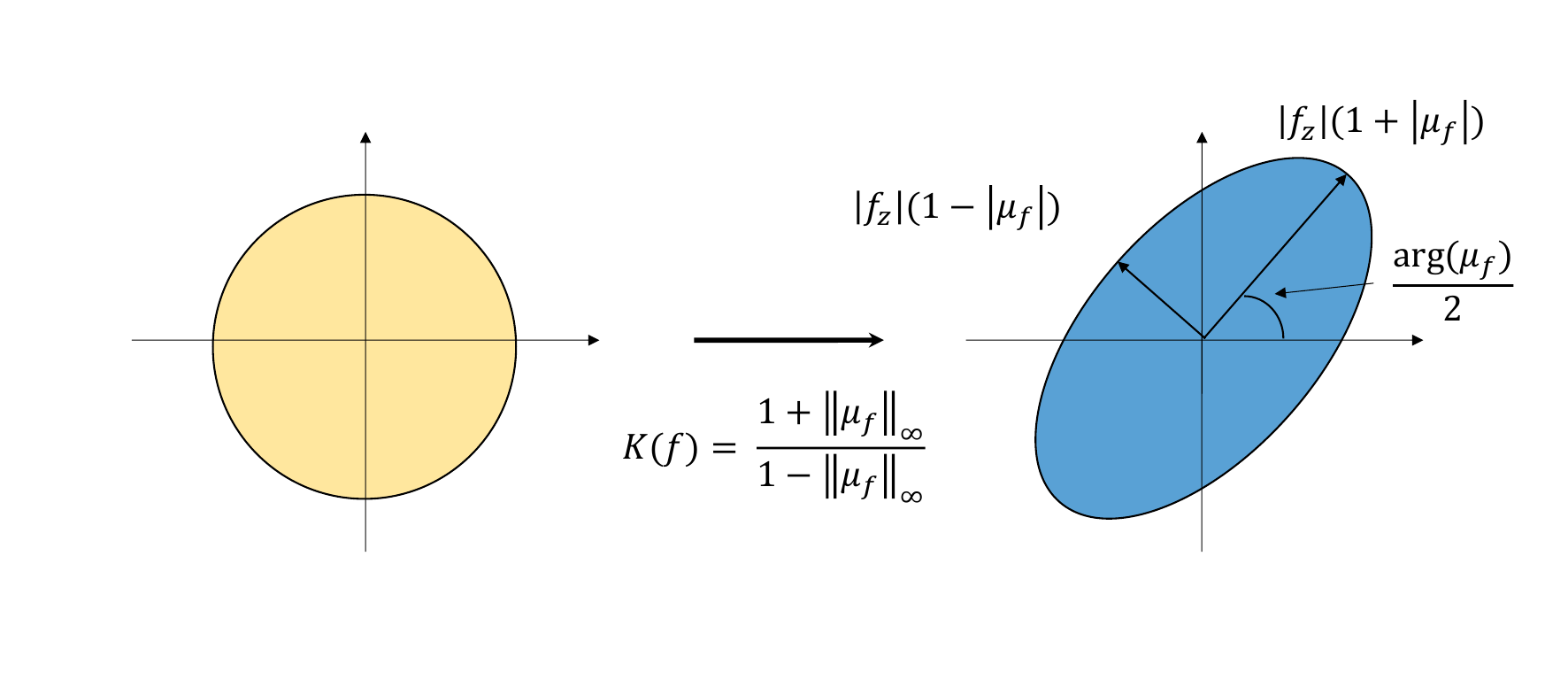}
    \caption{\textbf{An illustration of quasi-conformal maps.} Under a quasi-conformal map, an infinitesimal circle is mapped to an infinitesimal ellipse with bounded eccentricity.}
    \label{fig:quasiconformal_map}
\end{figure}

In particular, the complex-valued function $\mu$ is called the \textit{Beltrami coefficient} of $f$, which serves as an effective measure of the conformal distortion of $f$. Specifically, when $\mu = 0$, Eq.~\eqref{eqt:Beltrami_eq} degenerates to the Cauchy--Riemann equation~\eqref{eqt:cauchyriemann_z}, and hence the mapping $f$ is conformal. Intuitively, around a point $z_0 \in \mathbb C$, $f$ can be expressed as:
\begin{equation}\label{eqt:first_order_approximation}
    f(z) \approx f(z_0) + f_{z}(z_0)(z-z_0) + f_{\Bar{z}}(z_0)\overline{(z-z_0)} = f(z_0) + f_{z}(z_0)(z-z_0 + \mu(z_0)\overline{(z-z_0)}).
\end{equation}
The above formula suggests that $f$ maps an infinitesimal circle centered at $z_0$ to an infinitesimal ellipse centered at $f(z_0)$ (see Fig.~\ref{fig:quasiconformal_map}). Moreover, it is easy to see that when $\mu(z_0) = 0$, $f$ is conformal, and the infinitesimal circle remains an infinitesimal circle after mapping. Additionally, we can determine the angles at which the magnification and shrinkage are maximized, as well as quantify the degree of magnification and shrinkage at those angles. More specifically, the angle of maximal magnification is $\operatorname{arg}(\mu(z_0))/2$ with the magnifying factor $|f_{z}(z_0)|(1+|\mu(z_0)|)$, and the angle of maximal shrinkage is $(\operatorname{arg}(\mu(z_0))+\pi)/2$ with the shrinking factor $|f_{z}(z_0)(|1-|\mu(z_0)|)$. The maximal dilation of $f$ is 
\begin{equation}
    K(f) = \frac{1+\|\mu \|_{\infty}}{1-\|\mu\|_{\infty}}.
\end{equation}
Thus, the Beltrami coefficient $\mu$ encodes important geometric information about the quasi-conformal map $f$.

Besides encoding important geometric information about the associated quasi-conformal map, the Beltrami coefficient is also closely related to the bijectivity of the mapping. More specifically, we have the following result~\cite{gardiner2000quasiconformal}: 
\begin{theorem}
    If $f$ is a $C^1$ mapping satisfying $\|\mu \|_{\infty} <1$, then $f$ is locally invertible.
\end{theorem}

Moreover, suppose $f:\Omega_1 \rightarrow \Omega_2$ and $g:\Omega_2 \rightarrow \Omega_3$ are two quasi-conformal maps with Beltrami coefficients $\mu_f$ and $\mu_g$ respectively. The Beltrami coefficient of the composition map $g\circ f: \Omega_1 \rightarrow \Omega_3$ can be obtained by
\begin{equation}\label{composition_BC}
    \mu_{g \circ f} = \dfrac{\mu_f + (\mu_{g}\circ f) \tau }{1 + \overline{\mu_f}(\mu_g \circ f) \tau},
\end{equation}
where $\tau = \overline{f_z} / f_z$. As a special case, if we consider a conformal map $g$, then it is easy to see that $\mu_{g \circ f} = \mu_f$. In other words, composing a conformal map with a given quasi-conformal map will not change its Beltrami coefficient. 

More generally, given two Riemann surfaces $\mathcal{M}$ and $\mathcal{N}$, we can consider a quasi-conformal map $f:\mathcal{M} \to \mathcal{N}$ and assess the quasi-conformal distortion using the Beltrami differentials via the local charts of the surfaces. Specifically, a Beltrami differential $\mu(z) \frac{d\overline{z}}{dz}$ is an assignment to each chart $(U_{\alpha}, \phi_{\alpha})$ of an $L_{\infty}$ complex-valued function $\mu_{\alpha}$, defined on local parameter $z_{\alpha}$, such that
\begin{equation}
\mu_{\alpha}\frac{d\overline{z}_{\alpha}}{dz_{\alpha}} = \mu_{\beta}\frac{d\overline{z}_{\beta}}{dz_{\beta}}
\end{equation}
on the domain also covered by another chart $(U_{\beta}, \phi_{\beta})$, with $\frac{{d\overline{z}_{\beta}}}{dz_{\beta}} = \frac{d}{dz_{\alpha}} \phi_{\alpha \beta}$ and $\phi_{\alpha \beta} = \phi_{\beta} \circ \phi_{\alpha}^{-1}$. In practice, to assess the quasi-conformal distortion of a small local region under the mapping, we can conformally parameterize it onto the plane and apply the above-mentioned planar methods.

\section{Theoretical relationship between the Beltrami coefficient and angular distortion} \label{sec:main}

In this section, we introduce our established theoretical results on the relationship between the Beltrami coefficient and angular distortion on geometric mappings in detail.

\subsection{The relationship between angles in terms of the Beltrami coefficient}

First, we prove that the Beltrami coefficient can be used for expressing the change of an angle at a point under a quasi-conformal mapping:
\begin{theorem}\label{theorem1}
     Let $f$ be a quasi-conformal map with Beltrami coefficient $\mu$ at the point $z$. Let $\theta$ be the angle between a curve through $z$ and the maximal stretch direction (i.e., the direction along maximal magnification). Suppose that $f$ maps the angle $\theta$ to the angle $\phi$. Then, we have
\begin{equation}
    \tan \phi = \frac{1}{K_z} \tan \theta,
\end{equation}
where $K_z = \frac{1 + |\mu|}{1 - |\mu|}$ is the dilation of $f$.   
\end{theorem}

\textit{Proof.}
We prove the statement by aligning the principal directions (i.e., the directions of maximal magnification and shrinkage) with the coordinate axes via rotation and analyzing the affine approximation.

Specifically, since $f$ is a quasi-conformal mapping, it is differentiable at the point $z$. In the local neighborhood of $z$, the mapping $f$ can be approximated by an affine transformation:
\begin{equation}
w = A z + B \bar{z}, \quad A, B \in \mathbb{C},
\end{equation}
where $K_z= \dfrac{|A| + |B|}{|A| - |B|}$. 

\begin{figure}[t!]
    \centering
    \includegraphics[width=\textwidth]{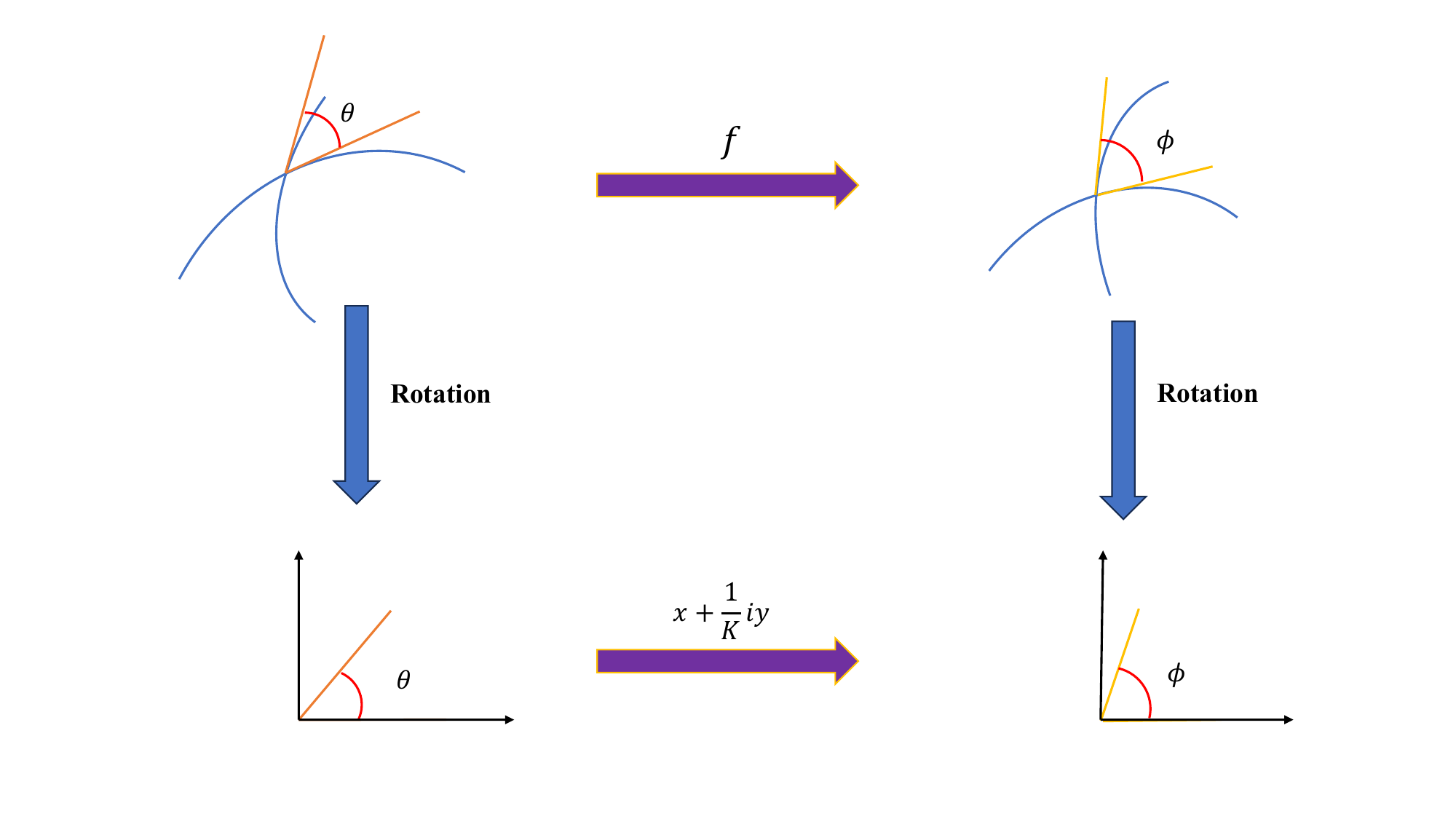}
    \caption{\textbf{An illustration for the proof of Theorem~\ref{theorem1}.} After a rotation, the angle induced by two curves is unchanged. Then we use an affine transformation to approximate the quasi-conformal mapping and obtain the relationship between the angles in terms of the Beltrami coefficient.}
    \label{fig:proof_mu}
\end{figure}

Now, we rotate the coordinate systems to align the principal directions with the coordinate axes (see Fig.~\ref{fig:proof_mu} as an illustration). Let $\arg A = \theta_A$ and $\arg B = \theta_B$. The rotation angles are given by:
\begin{align}
\alpha &= \frac{\theta_B - \theta_A}{2}, \\
\beta &= \frac{\theta_A + \theta_B}{2},
\end{align}
and the rotated new coordinates are:
\begin{align}
\zeta &= e^{-i\alpha} z, \\
\omega &= e^{-i\beta} w.
\end{align}
After this rotation, the affine transformation becomes:
\begin{equation}
\omega = \lambda_x \xi + i \lambda_y \eta, \quad \lambda_x, \lambda_y \in \mathbb{R}^+,
\end{equation}
where $\xi = \operatorname{Re}(\zeta)$, $\eta = \operatorname{Im}(\zeta)$, and:
\begin{align}
\lambda_x &= |A| + |B|, \\
\lambda_y &= |A| - |B|, \\
K_z &= \frac{\lambda_x}{\lambda_y}.
\end{align}
In the new coordinates, the principal directions now align with the real and imaginary axes in the \(\zeta\)-plane and \(\omega\)-plane. 

Now, we consider a ray in the \(\zeta\)-plane forming angle \(\theta\) with the principal direction (real axis). Its direction vector is \((\cos\theta, \sin\theta)\). After the affine transformation, the direction vector will be mapped to $(\lambda_x \cos\theta,  \lambda_y \sin\theta)$.

Let \(\phi\) be the angle between the transformed vector and the principal direction (real axis in the \(\omega\)-plane). Based on the above formulas, this angle satisfies:
\begin{equation}
\tan\phi = \frac{\lambda_y \sin\theta}{\lambda_x \cos\theta} = \frac{\lambda_y}{\lambda_x} \tan\theta.
\end{equation}
Substituting \(K_z = \lambda_x / \lambda_y\) gives:
\begin{equation}
\tan\phi = \frac{1}{K_z} \tan\theta. 
\end{equation}

Since the rotation of coordinates is angle-preserving, the above relationship holds in the original coordinates between the curve and the maximal stretch direction. \hfill $\blacksquare$

We remark that in the above proof, we assume that one side of the angle aligns with the maximal stretch direction. If neither side is along the maximal stretch direction, we can decompose the angle according to the maximal stretch direction and then obtain some related formulas. Specifically, we assume that the two sides make angles $\alpha$ and $\beta$ with the maximal stretch direction. The original angle is $\theta = \alpha - \beta$. (Here, we assume that the two sides are on the same side of the maximal stretch direction and $\alpha > \beta$.) By Theorem~\ref{theorem1}, the corresponding angles $\alpha^{\prime}$ and $\beta^{\prime}$ after mapping satisfy $\tan \alpha^{\prime} = \frac{1}{K_z}\tan \alpha$ and $\tan \beta^{\prime} = \frac{1}{K_z}\tan \beta$. The image angle $\phi = \alpha^{\prime} - \beta^{\prime}$ satisfies
\begin{equation}
    \tan \phi = \tan (\alpha^{\prime} - \beta^{\prime}) = \frac{\tan \alpha^{\prime} - \tan \beta^{\prime} }{1 + \tan \alpha^{\prime} \tan \beta^{\prime}} =K_z \frac{\tan \alpha - \tan \beta}{K_z^2 + \tan \alpha \tan \beta}.
\end{equation}
It can be observed that the derived formula is different from the previous one.

\subsection{Angular distortion estimation using the Beltrami coefficient}
Next, we prove that for a given quasi-conformal map $f$, an angle $\theta$ has the largest angular distortion if and only if its angle bisector is along one of the principal axes.

\begin{theorem}\label{theorem2}
    Let $f$ be a quasi-conformal map with Beltrami coefficient $\mu$, and let $\theta$ be an angle intersected by two curves. Denote $\phi$ as the corresponding angle after employing $f$. For a fixed $\theta$, the distortion $|\phi - \theta|$ achieves its maximum precisely when the bisector of $\theta$ is aligned with a principal axis.
\end{theorem}

\textit{Proof.}
Note that from Eq.~\eqref{eqt:first_order_approximation}, a quasi-conformal map can be locally approximated by an affine transformation. Moreover, by rotating the coordinate system, we may align the direction of maximal stretch with the $x$-axis. In this coordinate system, the mapping $f$ can be written as
\begin{equation}
    f(x,y) = (x, m y),
\end{equation}
where \( m = 1/K \) and \( K = \frac{1+|\mu|}{1-|\mu|} \ge 1 \) is the maximal dilation.

\begin{figure}[t!]
    \centering
    \includegraphics[width=\textwidth]{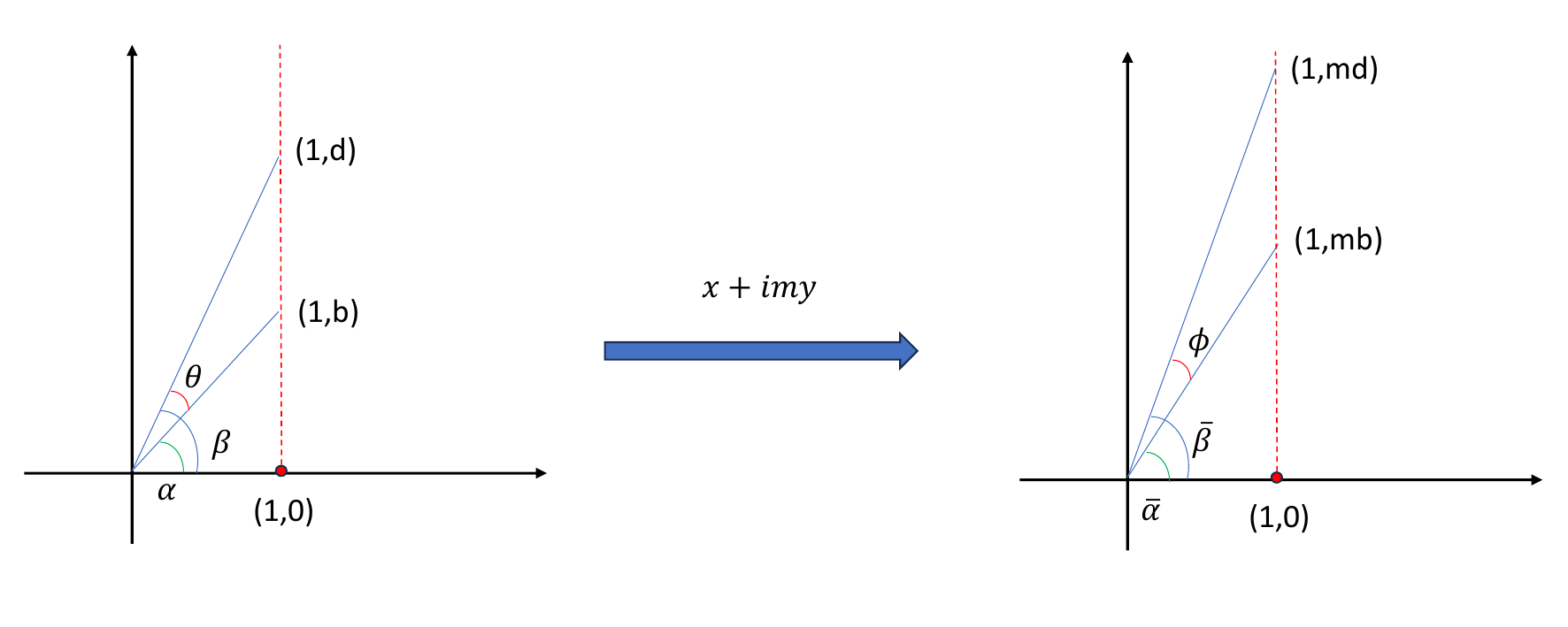}
    \caption{\textbf{An illustration of angular distortion.} Here, we show the angles $\theta$ and $\phi$ in their local neighborhood. (Left) The original angle $\theta$ obtained by the intersection of curves. (Right)~The result angle $\phi$ after the affine transformation.}
    \label{fig:BC_angle}
\end{figure}

Let the original angle be $\theta$ and set $n = \tan \theta$. Let $\alpha$ be the signed angle from the $x-$axis (maximal stretch direction) to one side of the angle, and write $b = \tan \alpha$. And the other side makes an angle $\beta = \alpha + \theta$ with the $x-$axis. After mapping, those two angles are mapped to $\Bar{\alpha}$ and $\Bar{\beta}$, respectively. And the result angle $\phi$ satisfies $\phi = \Bar{\beta} - \Bar{\alpha}$. Then, the problem can be transformed into the problem of finding the value $b$ to achieve the maximum of $|\phi - \theta|$ (see Fig.~\ref{fig:BC_angle}).

The tangent value of $\beta$ can be computed by:
\begin{equation}
    d = \tan \beta = \tan(\theta + \alpha) = \frac{\tan \alpha + \tan \theta}{1 - \tan \alpha \tan \theta} = \frac{b + n}{1 - bn}. 
\end{equation}
After applying the affine map, the corresponding angle $\phi$ satisfies
\begin{equation}
    \tan \phi = \tan(\Bar{\beta} - \Bar{\alpha}) = \frac{\tan \Bar{\beta} - \tan \Bar{\alpha}}{1 + \tan \Bar{\alpha} \tan \Bar{\beta}} = \frac{m(d-b)}{1+m^2db}.
\end{equation}

Substitute $d = \frac{b + n}{1 - bn}$, we can get:
\begin{equation}
    \tan \phi = \frac{mn(b^2 + 1)}{m^2b^2 + (m^2-1)nb + 1}.
\end{equation}
Next, by considering $\tan \phi$ as a function of $b\in \mathbb R$, we compute the derivative of 
\begin{equation}
    g(b) = \frac{mn(b^2 + 1)}{m^2b^2 + (m^2-1)nb + 1}.
\end{equation}
We have
\begin{equation}
    g'(b) = \frac{m n (m^{2} - 1) (n b^{2} - 2b - n)}
               {\bigl(m^{2} b^{2} + (m^{2} - 1) n b + 1\bigr)^{2}}.
\end{equation}
The critical points satisfy $nb^2 - 2b - n = 0$, with solutions
\begin{equation}
    b_{1,2} = \frac{1 \pm \sqrt{1+n^2}}{n}.
\end{equation}
Sign analysis of $f'(b)$, noting that $m^2-1 \leq 0$, confirms that $f(b)$ decreases on $(-\infty,b_1)$, increases on $(b_1,b_2)$, and decreases again on $(b_2,\infty)$. Hence $b_1$ yields the global minimum $\phi_{\min}$ and $b_2$ yields the global maximum $\phi_{\max}$ of $\phi$. 

Using the identities $b_1 = -\tan(\theta/2)$ and $b_2 = \cot(\theta/2)$, we can verify that:
\begin{itemize}
    \item $b_1$ corresponds to the bisector aligned with the $x$-axis (maximal stretch),
    \item $b_2$ corresponds to the bisector aligned with the $y$-axis (minimal stretch).
\end{itemize}

Therefore, the distortion $|\phi - \theta|$ achieves its maximum when the bisector lies along one of the two principal axes, and is strictly smaller otherwise.

\hfill $\blacksquare$

\subsection{Upper bound of the angular distortion in terms of the Beltrami coefficient}
In the following, we consider the largest angular distortion at a point under a quasi-conformal map. We assume that the bisector is aligned with the maximal stretch direction. When the bisector is aligned with the minimal stretch direction, the maximum expansion is symmetric and yields the same bound. Therefore, we omit that case.

For simplicity, we denote the original half-angle as $\theta$ (and hence the full angle is $2\theta$) and let $\phi$ be the corresponding half-angle after the mapping. Since the local approximation of $f$ by an affine transformation is $(x,y) \mapsto (x, m y)$ with $m = 1/K$ as discussed previously, we have
\begin{equation}
\tan \phi = \frac{1}{K} \tan \theta.
\end{equation}
Now, we define the half-angle deviation $\delta = \theta - \phi$. We have
\begin{equation}
\delta = \theta - \arctan \left( \frac{1}{K} \tan \theta \right).
\end{equation}
To obtain the maximum of $\delta$, we differentiate $\delta$ with respect to $\theta$ and set the derivative to zero:
\begin{equation}
\frac{d\delta}{d\theta} = 1 - \frac{\frac{1}{K} \sec^{2}\theta}
                                 {1 + \bigl( \frac{1}{K} \tan \theta \bigr)^{2}} = 0.
\end{equation}
Then we simplify the above derivative equation and obtain
\begin{equation}
    K\sec^{2}\theta = K^2 + \tan^{2}\theta.
\end{equation}
Using the identity $\sec^{2}\theta = 1 + \tan^{2}\theta$, we have
\begin{equation}
    (\tan^{2}\theta)(K - 1) = K(K - 1).
\end{equation}
Since $K \neq 1$, the above equation yields
\begin{equation}
    \tan^{2}\theta = K.
\end{equation}
Taking $\theta \in (0,\pi/2)$, we have $\tan \theta = \sqrt{K}$ and consequently $\tan \phi = 1/\sqrt{K}$. Therefore, the maximal half-angle deviation $\delta_{\max}$  satisfies
\begin{equation}
\tan \delta_{\max} = \frac{\tan \theta - \tan \phi}
                            {1 + \tan \theta \tan \phi}
                   = \frac{\sqrt{K} - 1/\sqrt{K}}{1 + 1}
                   = \frac{K - 1}{2\sqrt{K}}.
\end{equation}

Constructing a right triangle with opposite side $K-1$ and adjacent side $2\sqrt{K}$, we get the hypotenuse as follows:
\begin{equation}
    \text{hypotenuse} = \sqrt{(K-1)^2 + (2\sqrt{K})^2} = K+1.
\end{equation}
Thus, we have
\begin{equation}
    \cos \delta_{\max} = \frac{2\sqrt{K}}{K+1} \quad \Rightarrow \quad \delta_{\max} = \arccos\left(\frac{2\sqrt{K}}{K+1}\right) = \arcsin\left(\frac{K-1}{K+1}\right) = \arcsin\left(|\mu|\right).
\end{equation}

Recalling that the Beltrami coefficient satisfies $|\mu| = \frac{K-1}{K+1}$, we obtain
\begin{equation}
\delta_{\max}(z) = \arcsin(|\mu(z)|).
\end{equation}
Since the original full angle is $2\theta$, the maximal distortion of the full angle equals $2\delta_{\max}$. In other words, the maximal angular distortion at the point $z$ is 
\begin{equation}
    \epsilon_{\mu}(z) = 2 \arcsin(|\mu(z)|).
\end{equation}

\subsection{The discrete case}
After establishing the above theoretical results and bounds, we can move on to the discrete case. More specifically, note that in computer graphics and geometry processing, planar shapes and surfaces are commonly discretized in the form of triangle meshes. Therefore, it is natural to consider how the discrete Beltrami coefficient is related to the angles of the triangle elements.

For any given quasi-conformal mapping $f = u + iv$, its corresponding Beltrami coefficient $\mu$ can be computed as following
\begin{equation}
    \mu = \frac{(u_x - v_y) + i(v_x + u_y)}{(u_x + v_y) + i (v_x - u_y)}.
\end{equation}
In the discrete case, the mapping $f$ on the triangle mesh is piecewise linear, and the corresponding Beltrami coefficient $\mu$ can also be discretized on each triangle $T$. The details are provided in~\cite{lui2013texture}, and we outline the key formulations as follows. Consider a piecewise linear function $f$. Within the triangle $T$, it can be written as
\begin{equation}
    f|_{T}(x, y) = \begin{pmatrix}
a_{T} x + b_{T} y + p \\
c_{T} x + d_{T} y + q
\end{pmatrix},
\end{equation}
where $a_{T}$, $b_{T}$, $c_{T}$, $d_{T}$,  $p$, and $q$ are constants. Following the Beltrami equation, the associated Beltrami coefficient $\mu_T$ can be computed by:
\begin{equation} \label{eqt:mu_T}
    \mu_T = \frac{\partial f / \partial \bar{z}}{\partial f / \partial z} = \frac{(a_{T} - d_{T}) + i(c_{T} + b_{T})}{(a_{T} + d_{T}) + i(c_{T} - b_{T})}.
\end{equation}
It is clear that $\mu_T$ is a constant for each triangle $T$. Now, with the discrete Beltrami coefficient $\mu_T$ defined on every $T$, we can consider the Beltrami-coefficient-based estimate in the previous section as
\begin{equation} \label{eqt:epsilon_mu_T}
    \epsilon_{\mu_T} = 2 \arcsin(|\mu_T|)
\end{equation}
for every face $T$.

We remark that while the Beltrami coefficient is defined based on mappings on the complex plane, we can easily extend the computation of the Beltrami coefficient for surface mappings. Let $f: \mathcal{M} \rightarrow \mathcal{N}$ be a quasi-conformal mapping from a surface $\mathcal{M}$ to another domain $\mathbb{N}$, where $\mathcal{M}$ and $\mathcal{N}$ can be in either $\mathbb{R}^3$ or $\mathbb{R}^2$. For each triangle element $T_0$ in $\mathcal{M}$ and its corresponding triangle $T_1$ in $\mathcal{N}$, we can first map both of them onto the complex plane using a rigid transformation. We denote the resulting triangles as $\Tilde{T_0}$ and $\Tilde{T_1}$. Now, to compute the Beltrami coefficient of the overall mapping $f$ on the face pair $T_0, T_1$, we can instead consider the mapping $\Tilde{f}: \Tilde{T_0} \rightarrow \Tilde{T_1}$ between the two transformed triangles on the complex plane, for which the Beltrami coefficient can be computed easily (note that here we only focus on the norm of it). As rigid transformations are conformal, by Eq.~\eqref{composition_BC}, the above operations will not change the norm of the Beltrami coefficient. We can therefore obtain both $|\mu_T|$ (using Eq.~\eqref{eqt:mu_T}) and $\epsilon_{\mu_T}$ (using Eq.~\eqref{eqt:epsilon_mu_T}) without encountering any issue.

As for the discrete angular distortion, from the triangulation of the given surface, we can easily define the absolute angular distortion at every angle as
\begin{equation}
    \epsilon_{\text{angle}}([v_i, v_j, v_k]) = |\angle [f(v_i), f(v_j), f(v_k)] - \angle [v_i, v_j, v_k]|,
\end{equation}
where the three points $v_i, v_j, v_k$ define an angle, $\angle [v_i, v_j, v_k]$ is the angle value on the original surface, and $\angle [f(v_i), f(v_j), f(v_k)]$ is the angle value under the mapping $f$.

Moreover, as the discrete Beltrami coefficient $\mu_T$ is a face-based quantity, to ensure a fair comparison, we can also construct a face-based angular distortion measure by taking the average value of the angular distortions of the three angles in every triangle $T$. In other words, for every triangle $T$ with vertices $v_i, v_j, v_k$, we define the face-based angular distortion measure as
\begin{equation}
    \epsilon_{\text{angle}_T} = \frac{1}{3}\left(\epsilon_{\text{angle}}([v_i, v_j, v_k]) +\epsilon_{\text{angle}}([v_j, v_k, v_i]) + \epsilon_{\text{angle}}([v_k, v_i, v_j])\right). 
\end{equation}

\section{Experimental results} \label{sec:experiment}
In this section, we present experimental results to verify our theoretical results and bounds on the Beltrami coefficient and angular distortion. The experiments were performed using MATLAB R2021a on the Windows platform on a computer with an Intel(R) Core(TM) i9-12900 2.40 GHz processor and 32 GB memory. All surfaces are discretized in the form of triangle meshes.

\begin{figure}[t!]
    \centering
    \includegraphics[width=\textwidth]{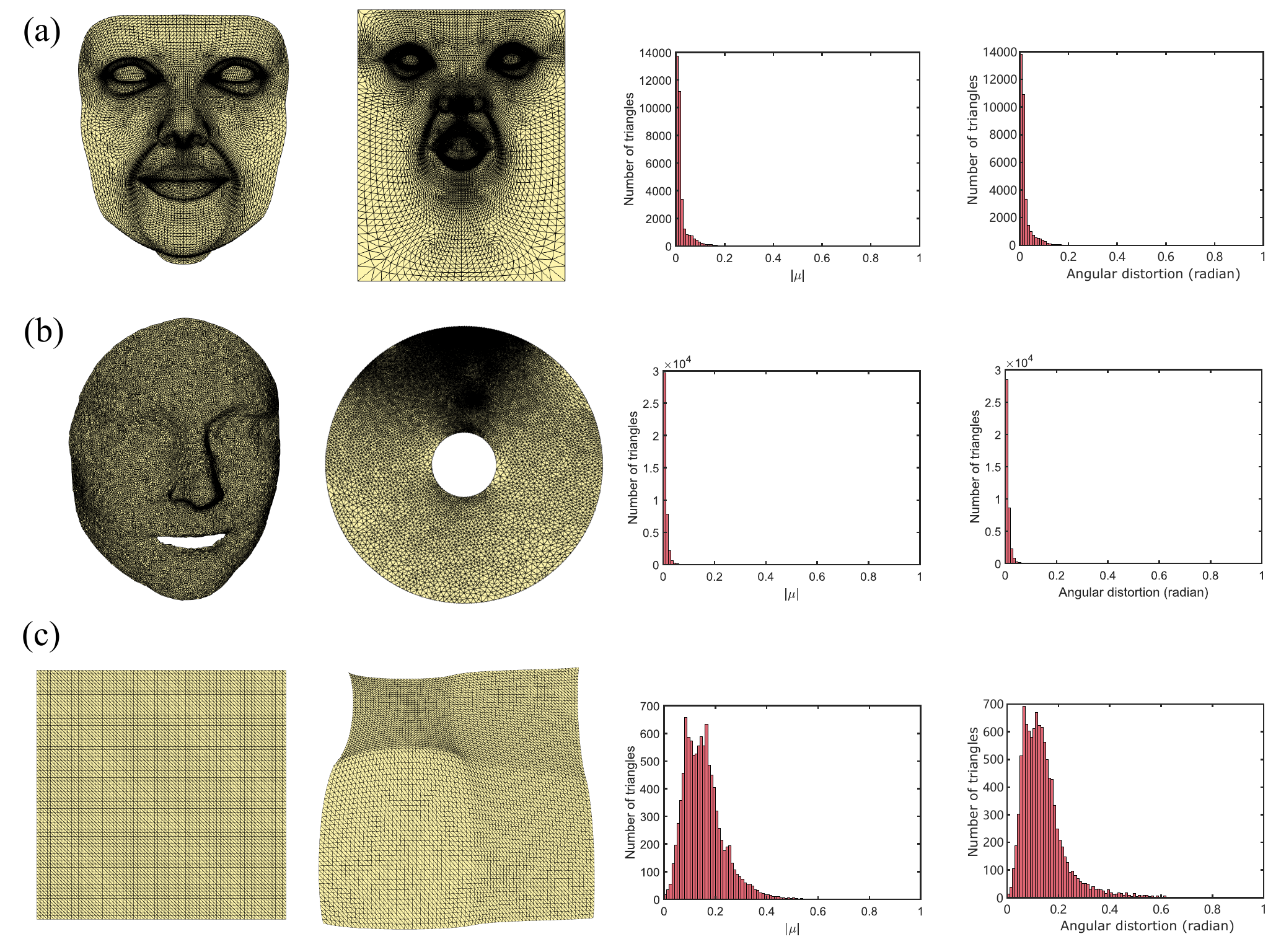}
     \caption{\textbf{Evaluation of the Beltrami coefficient and angular distortion of different mappings for open surfaces.} Each row shows one example. (a)~Rectangle conformal map~\cite{meng2016tempo} of the Human face model. (b)~Annulus conformal map~\cite{choi2021efficient} of the Sophie model. (c)~Density-equalizing map~\cite{choi2018density} of the Square model. Left to right: The input surface mesh, the mapping result, the histogram of the norm of the Beltrami coefficient $|\mu_T|$, and the histogram of the face-based angular distortion $\epsilon_{\text{angle}_T}$.}
    \label{fig:open_surface}
\end{figure}

\subsection{Simply-connected and multiply-connected open surfaces}
We start by considering different mapping methods for simply-connected and multiply-connected open surfaces.
In Fig.~\ref{fig:open_surface}, we present three mapping examples, including a rectangular conformal map of a simply-connected open human face model using~\cite{meng2016tempo}, an annulus conformal map of a punctured face example using~\cite{choi2021efficient}, and a density-equalizing map of a simply-connected open Square model using \cite{choi2018density}. For each example, we plot the histogram of the norm of the discretized Beltrami coefficient $|\mu_T|$ and the histogram of face-averaged angular distortion $\epsilon_{\text{angle}_T}$. It can be observed that for the two conformal mapping examples, the angular distortion histogram closely resembles the one for the Beltrami coefficient. For the density-equalizing mapping example, the angular distortion histogram is very similar to that of the Beltrami coefficient. The above results demonstrate the close relationship between the Beltrami coefficient and angular distortions.

\begin{table}[t]
\centering
\small
\caption{\textbf{Quantitative analysis of the Beltrami coefficient and angular distortion of geometric mappings for simply-connected and multiply-connected open surfaces.} For each surface mapping example, we record the number of triangle elements, the mean and maximum values of $|\mu_T|$, the mean and maximum values of $\epsilon_{\text{angle}_T}$ and $\epsilon_{\mu_T}$, and the maximum value of $\epsilon_{\text{angle}}$.} 
\label{tab:open_surface}
    \resizebox{1\linewidth}{!}{
\begin{tabular}{C{22mm}| r r r r r r r r}
\toprule
Surface and mapping method & \# Faces & $\text{mean}(|\mu_T|)$ & $\text{max}(|\mu_T|)$ & \begin{tabular}{@{}c@{}}$\text{mean}(\epsilon_{\text{angle}_T})$\end{tabular} & \begin{tabular}{@{}c@{}}$\text{mean}(\epsilon_{\mu_T})$\end{tabular} & \begin{tabular}{@{}c@{}}$\text{max}(\epsilon_{\text{angle}_T})$\end{tabular}& \begin{tabular}{@{}c@{}}$\text{max}(\epsilon_{\text{angle}})$\end{tabular} & \begin{tabular}{@{}c@{}}$\text{max}(\epsilon_{\mu_T})$\end{tabular} \\
\midrule
Brain (disk conformal) & 48463 & 0.0251 & 0.4985 & 0.0254 & 0.0501&  0.5468 & 0.8202 & 1.0437 \\ \hline 
Human face (rectangular conformal) & 34144 & 0.0216 & 0.393 & 0.0212 & 0.0432 & 0.4544 & 0.5887 & 0.8078 \\ \hline 
Sophie (annulus conformal) & 41038 & 0.0096 & 0.6631 & 0.0102 & 0.0192 & 0.6129 & 0.8163 & 1.4499 \\ \hline 
Square (density-equalizing) & 10368 & 0.1546 & 0.7467 & 0.1410 & 0.3118 & 0.9682 & 1.4523 &  1.6863   \\ %
\bottomrule
\end{tabular}
}
\end{table}

For a more quantitative analysis, Table~\ref{tab:open_surface} records the number of triangle elements, the mean and maximum values of $|\mu_T|$, the mean and maximum values of $|\epsilon_{\text{angle}_T}|$ and $|\epsilon_{\mu_T}|$, and the maximum values of $|\epsilon_{\text{angle}}|$ for all mapping examples in Fig.~\ref{fig:brain_disk} and Fig.~\ref{fig:open_surface}. It can be observed that for all examples, the difference between $\text{mean}(|\mu_T|)$ and $\text{mean}(\epsilon_{\text{angle}_T})$ is very small. This shows that the two concepts are highly consistent and can both be used for assessing the overall conformality of the mapping. We can also consider the maximum value of $|\mu_T|$, $\epsilon_{\text{angle}_T}$, and $|\epsilon_{\text{angle}}|$. In this case, their actual values are more different, but the three sets of values still follow a consistent trend. As for our maximum angular distortion estimate $\epsilon_{\mu_T}$, it can be observed that both the mean and maximum values are consistently higher than those of $\epsilon_{\text{angle}_T}$ and $|\epsilon_{\text{angle}}|$. This can be explained by the fact that $\epsilon_{\mu_T}$ requires all angle bisectors to be aligned with the direction of maximum stretching and hence serves as an upper bound for the angular distortion. In all examples, we see that the differences among $\text{max}(\epsilon_{\text{angle}})$, $\text{max}(\epsilon_{\text{angle}_T})$, and $\text{max}(\epsilon_{\mu_T})$ are generally large, suggesting that the angle bisectors do not align with the direction of maximum stretching. Overall, the above results show that the Beltrami coefficient and angular distortion are highly related quantitatively.

\subsection{Genus-0 closed surfaces}

\begin{figure}[t!]
    \centering
    \includegraphics[width=\textwidth]{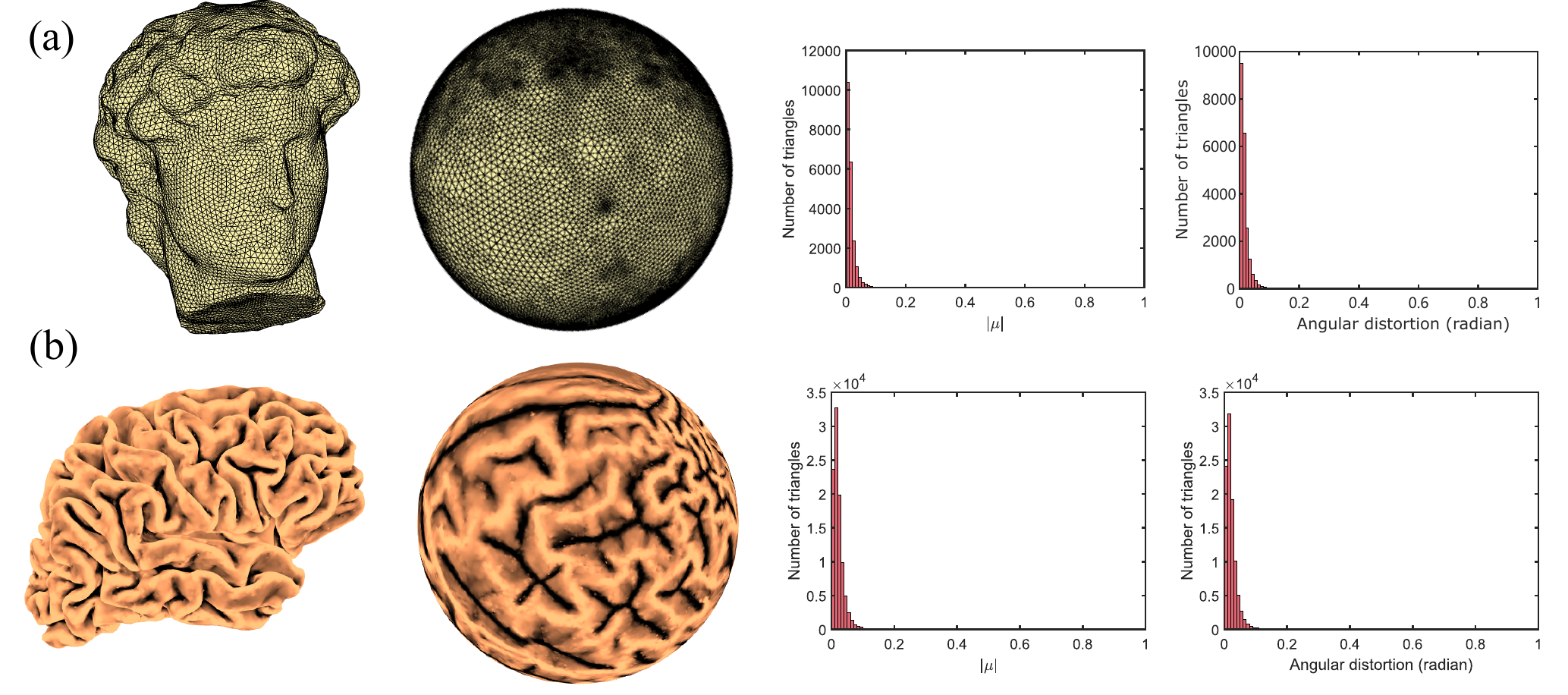}
    \caption{\textbf{Evaluation of the Beltrami coefficient and angular distortion of spherical conformal parameterization~\cite{choi2015flash} of genus-0 surfaces.} Each row shows one example. (a) The David model. (b) The brain model. Left to right: The input surface mesh, the parameterization result, the histogram of the norm of the Beltrami coefficient $|\mu_T|$, and the histogram of the face-based angular distortion $\epsilon_{\text{angle}_T}$.}
    \label{fig:sphere_conf}
\end{figure}

Now, we consider different mapping examples for genus-0 closed surfaces. Note that the Beltrami coefficient is a complex-valued function defined on the complex plane. To quantify the quasi-conformal distortions of the triangular elements under the mapping, we isometrically map each pair of the triangular elements before and after the mapping into the complex plane and compute the Beltrami coefficients using the planar states. First, Fig.~\ref{fig:sphere_conf} shows two spherical conformal mapping~\cite{choi2015flash} examples and the evaluation of the Beltrami coefficient and angular distortion. Again, we can see a high similarity between the distribution of $|\mu_T|$ and the face-based angular distortion $\epsilon_{\text{angle}_T}$. This shows that the Beltrami coefficient can also serve as an accurate tool to measure the angular distortion. 

In Fig.~\ref{fig:sphere_sdem}, we consider spherical area-preserving parameterization examples obtained by the spherical density-equalizing mapping method in~\cite{lyu2024spherical} to study the two types of conformality measures for highly non-conformal mappings. Since the mappings are area-preserving, it is expected that the conformal distortion is much larger. We observe that the histograms of $|\mu_T|$ and $\epsilon_{\text{angle}_T}$ are still highly consistent in this case, which suggests that both the Beltrami coefficient and the face-based angular distortion measures can effectively capture the conformal distortion of a wide range of mappings.

\begin{figure}[t!]
    \centering
    \includegraphics[width=\textwidth]{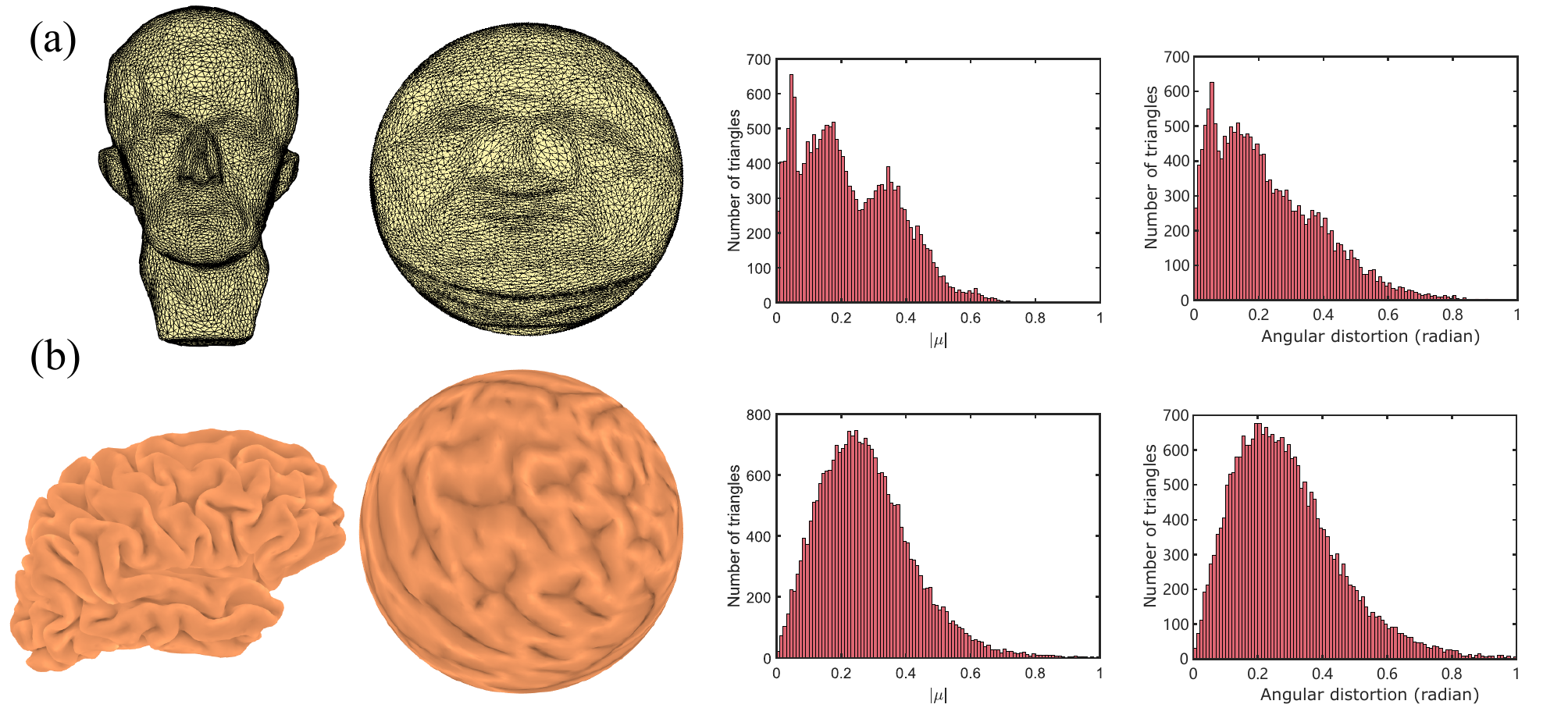}
    \caption{\textbf{Evaluation of the Beltrami coefficient and angular distortion of spherical density-equalizing map~\cite{lyu2024spherical} for genus-0 closed surfaces.} Each row shows one example. (a) The Max Planck model. (b) The brain model. Left to right: The input surface mesh, the parameterization result, the histogram of the norm of the Beltrami coefficient $|\mu_T|$, and the histogram of the face-based angular distortion $\epsilon_{\text{angle}_T}$.}
    \label{fig:sphere_sdem}
\end{figure}

\begin{figure}[t!]
    \centering
    \includegraphics[width=\textwidth]{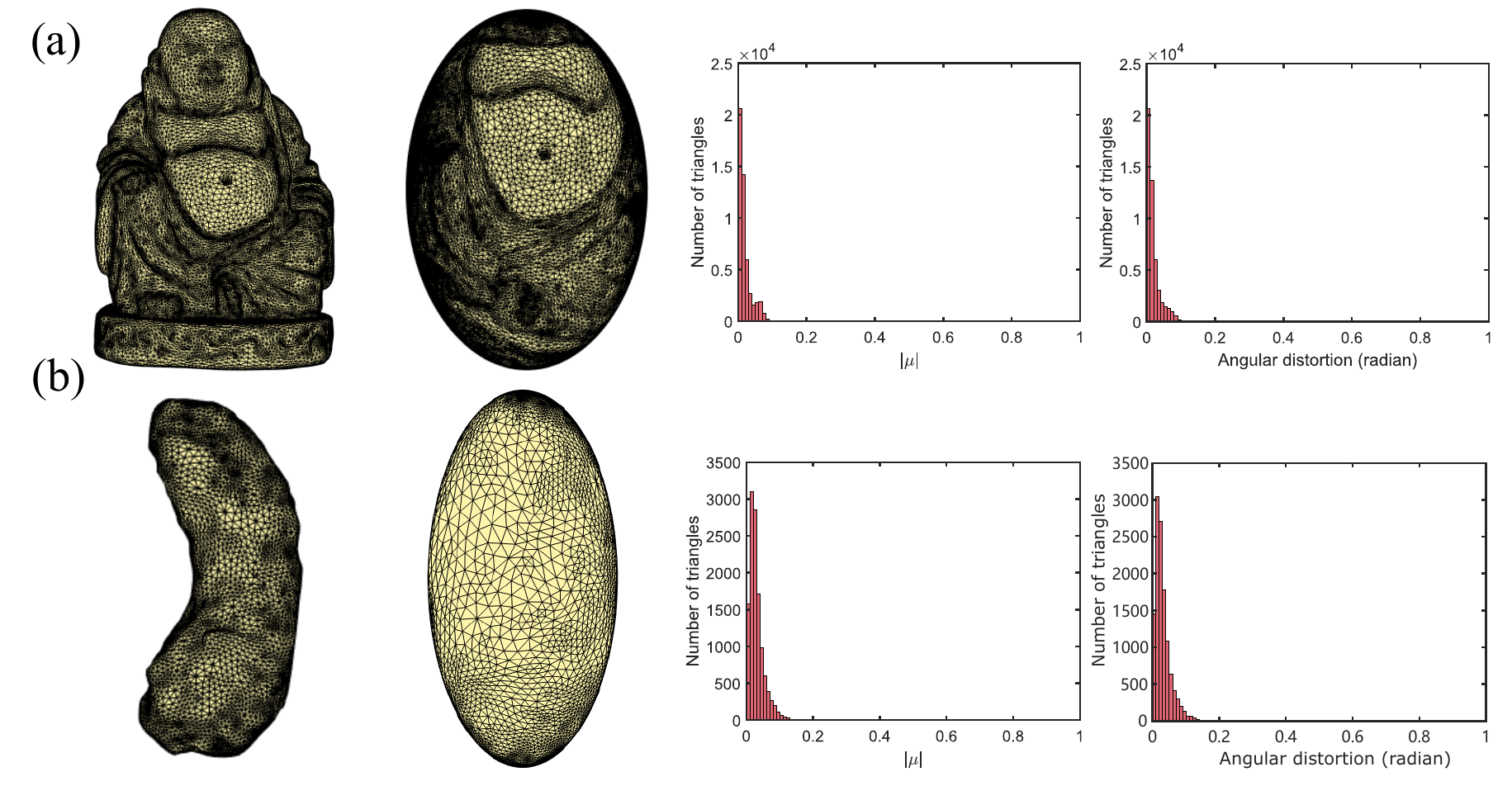}
    \caption{\textbf{Evaluation of the Beltrami coefficient and angular distortion of ellipsoidal conformal map~\cite{choi2024fast} for genus-0 closed surfaces.} Each row shows one example. (a) The Buddha model. (b) The hippocampus model. Left to right: The input surface mesh, the parameterization result, the histogram of the norm of the Beltrami coefficient $|\mu_T|$, and the histogram of the face-based angular distortion $\epsilon_{\text{angle}_T}$.}
    \label{fig:ellip_conf}
\end{figure}

\begin{figure}[t!]
    \centering
    \includegraphics[width=\textwidth]{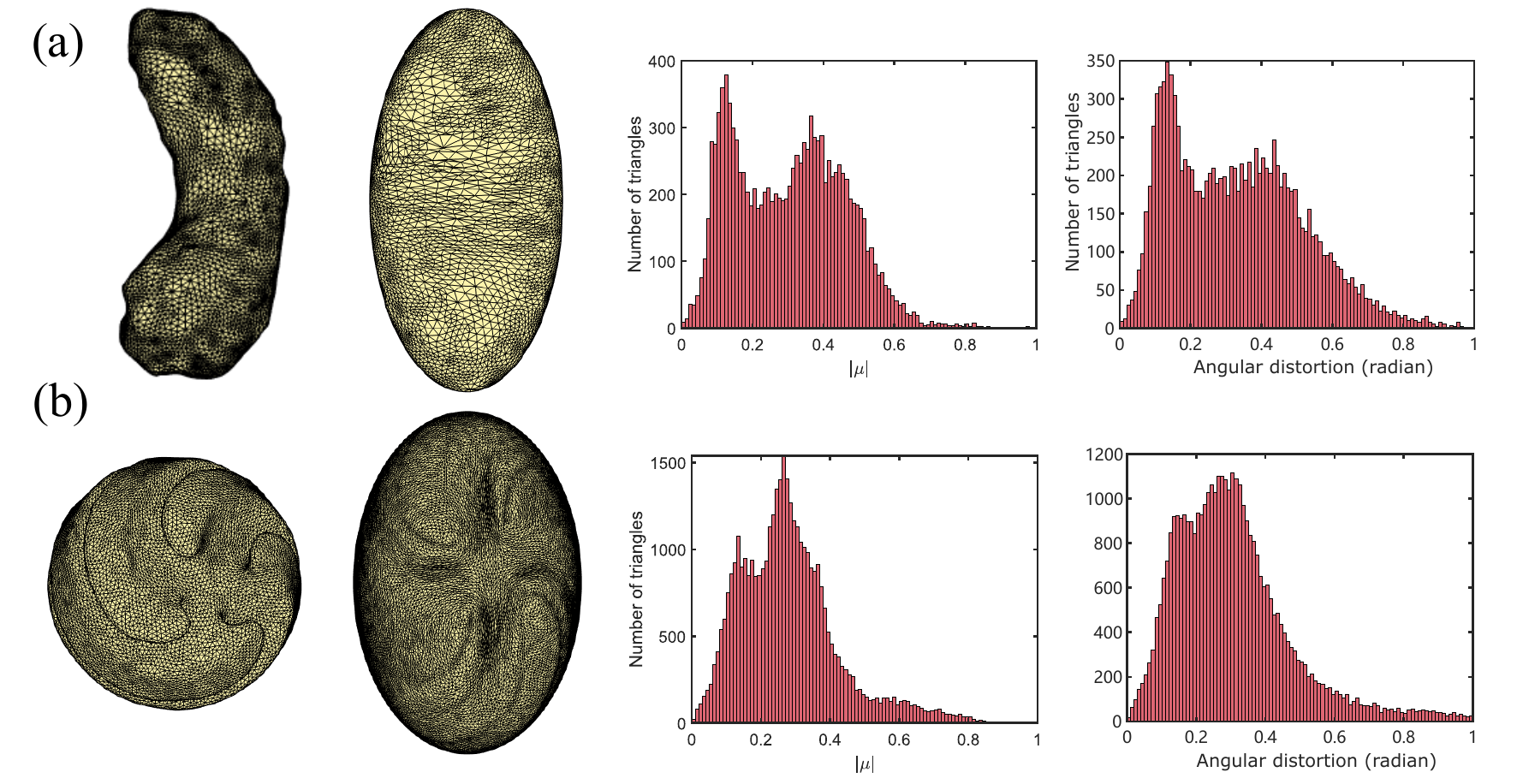}
    \caption{\textbf{Evaluation of the Beltrami coefficient and angular distortion of ellipsoidal density-equalizing map~\cite{lyu2024ellipsoidal} for genus-0 closed surfaces.} Each row shows one example. (a) The hippocampus model. (b) The twisted ball model. Left to right: The input surface mesh, the parameterization result, the histogram of the norm of the Beltrami coefficient $|\mu_T|$, and the histogram of the face-based angular distortion $\epsilon_{\text{angle}_T}$.}
    \label{fig:ellip_edem}
\end{figure}

In Fig.~\ref{fig:ellip_conf} and Fig.~\ref{fig:ellip_edem}, we further consider ellipsoidal mapping examples for genus-0 closed surfaces. Specifically, we apply the ellipsoidal conformal mapping method~\cite{choi2024fast} and the ellipsoidal area-preserving mapping~\cite{lyu2024ellipsoidal} on a variety of surfaces in graphics and medicine. Here, note that the target domain in the ellipsoidal mappings can be flexible and may involve ellipsoidal geometries with different prescribed radii, and hence they serve as good test cases for us to consider our measures for more general mapping problems. Analogous to the previous examples, we can see that the histograms of $|\mu_T|$ and $\epsilon_{\text{angle}_T}$ are highly consistent in all conformal and area-preserving examples.

Table~\ref{tab:closed_surface} presents the detailed statistics of the above-mentioned mapping examples for genus-0 closed surfaces. For each surface mapping example, we record the number of triangle elements, the mean and maximum values of $|\mu_T|$, the mean and maximum values of $|\epsilon_{\text{angle}_T}|$ and $|\epsilon_{\mu_T}|$, and the maximum value of $\epsilon_{\text{angle}}$. Again, the results reveal the close relationship between the Beltrami coefficient and the angular distortion. Specifically, the difference between $\text{mean}(|\mu_T|)$ and $\text{mean}(\epsilon_{\text{angle}_T})$ is again very small in all examples. Also, the trends in $\text{max}(|\mu_T|)$ and $\text{max}(\epsilon_{\text{angle}})$ are consistent. Moreover, the Beltrami coefficient-based estimate $\max(|\epsilon_{\mu_T}|)$ provides a rigorous upper bound for the maximum angular distortions $\max(|\epsilon_{\text{angle}_T}|)$ and $\max(|\epsilon_{\text{angle}}|)$. Overall, the experiment results have verified and illustrated the theoretical connection between the Beltrami coefficient and angular distortion established in our work.

\begin{table}[t!]
\centering
\small
\caption{\textbf{Quantitative analysis of the Beltrami coefficient and angular distortion of geometric mappings for genus-0 closed surfaces.} For each surface mapping example, we record the number of triangle elements, the mean and maximum values of $|\mu_T|$, the mean and maximum values of $\epsilon_{\text{angle}_T}$ and $\epsilon_{\mu_T}$, and the maximum value of $\epsilon_{\text{angle}}$.} 
\label{tab:closed_surface}
    \resizebox{1\linewidth}{!}{
\begin{tabular}{C{30mm}|  r r r r r r r r}
\toprule
Surface and mapping method  & \# Faces & $\text{mean}(|\mu_T|)$ & $\text{max}(|\mu_T|)$ & \begin{tabular}{@{}c@{}}$\text{mean}(\epsilon_{\text{angle}_T})$\end{tabular} & \begin{tabular}{@{}c@{}}$\text{mean}(\epsilon_{\mu_T})$\end{tabular} & \begin{tabular}{@{}c@{}}$\text{max}(\epsilon_{\text{angle}_T})$\end{tabular} & \begin{tabular}{@{}c@{}}$\text{max}(\epsilon_{\text{angle}})$\end{tabular} & \begin{tabular}{@{}c@{}}$\text{max}(\epsilon_{\mu_T})$\end{tabular} \\
\midrule
David (conformal) & 21338 & 0.0145 & 0.1513 & 0.0159 & 0.0291 & 0.1789 & 0.2684 & 0.3037 \\ \hline
Brain (conformal) & 96970 & 0.0217 & 0.7474 & 0.0220 & 0.0434 & 0.8711 & 1.3066 & 1.6884 \\ \hline
Max Planck (area-preserving) & 25000 & 0.2221 & 0.8622 & 0.2280 & 0.4479 &  0.9769  & 1.4654 & 2.0794 \\ \hline 
Brain (area-preserving) & 25000 & 0.2777 & 0.9977 & 0.2921 & 0.5828 & 1.7617 & 2.6425 & 3.0060 
\\ \hline
Buddha (conformal) & 50002 & 0.0191 & 0.2603 &	0.0195 & 0.0380 & 0.2525 & 0.3878 & 0.5266\\ \hline 
Hippocampus (conformal) & 12000 & 0.0299 & 0.2451 & 0.0312 & 0.0598 & 0.2830 & 0.4245 & 0.4952\\ \hline
Twisted ball (area-preserving) & 38620 & 0.2798 & 0.8828 & 0.3110 & 0.5672 & 1.3220 & 1.9830 & 2.1723 \\ \hline 
Hippocampus (area-preserving) & 12000 & 0.3079 & 0.9784 & 0.3336 & 0.6261 & 1.3347 & 2.0021 & 2.7247\\ 
\bottomrule
\end{tabular}
}
\end{table}

\subsection{Further analysis on mesh quality and mesh resolution}
It is natural to ask whether our established connection between the Beltrami coefficient and angular distortion may be affected by the quality and resolution of the discrete meshes. In this section, we further analyze the Beltrami coefficient and angular distortion of mapping examples with different mesh quality and mesh resolution. 

First, we consider two input meshes with different mesh quality. Specifically, in Fig.~\ref{fig:quality}(a), we consider a high-quality Delaunay triangulation (with around 900 triangles) on a rectangular domain. Then, in Fig.~\ref{fig:quality}(b), we use the same set of vertices but construct another triangulation that violates the Delaunay property and contains a large number of sharp triangles. We then apply the same geometric mapping to both meshes and assess the Beltrami coefficient and angular distortion of the two results. In both mapping examples, it can be observed that histograms of $|\mu_T|$ and $\epsilon_{\mu_T}$ are consistent. One can also see that the corresponding histograms in (a) and (b) are highly similar. This suggests that our method is robust and maintains consistent performance across meshes of different qualities.

Next, we consider mapping examples with different mesh resolutions as shown in Fig.~\ref{fig:resolution}. Specifically, in (a) we consider a coarse triangle mesh with around 900 elements on a square domain. Then, in (b) we consider a denser mesh with around 3600 elements. For both meshes, we employ the same geometric mapping and analyze the results. It can again be observed that both the $|\mu_T|$ and $\epsilon_{\mu_T}$ histograms remain highly similar in both examples. This suggests that our established results and bounds are valid and consistent for different mesh resolutions.

\begin{figure}[t!]
    \centering
    \includegraphics[width=\textwidth]{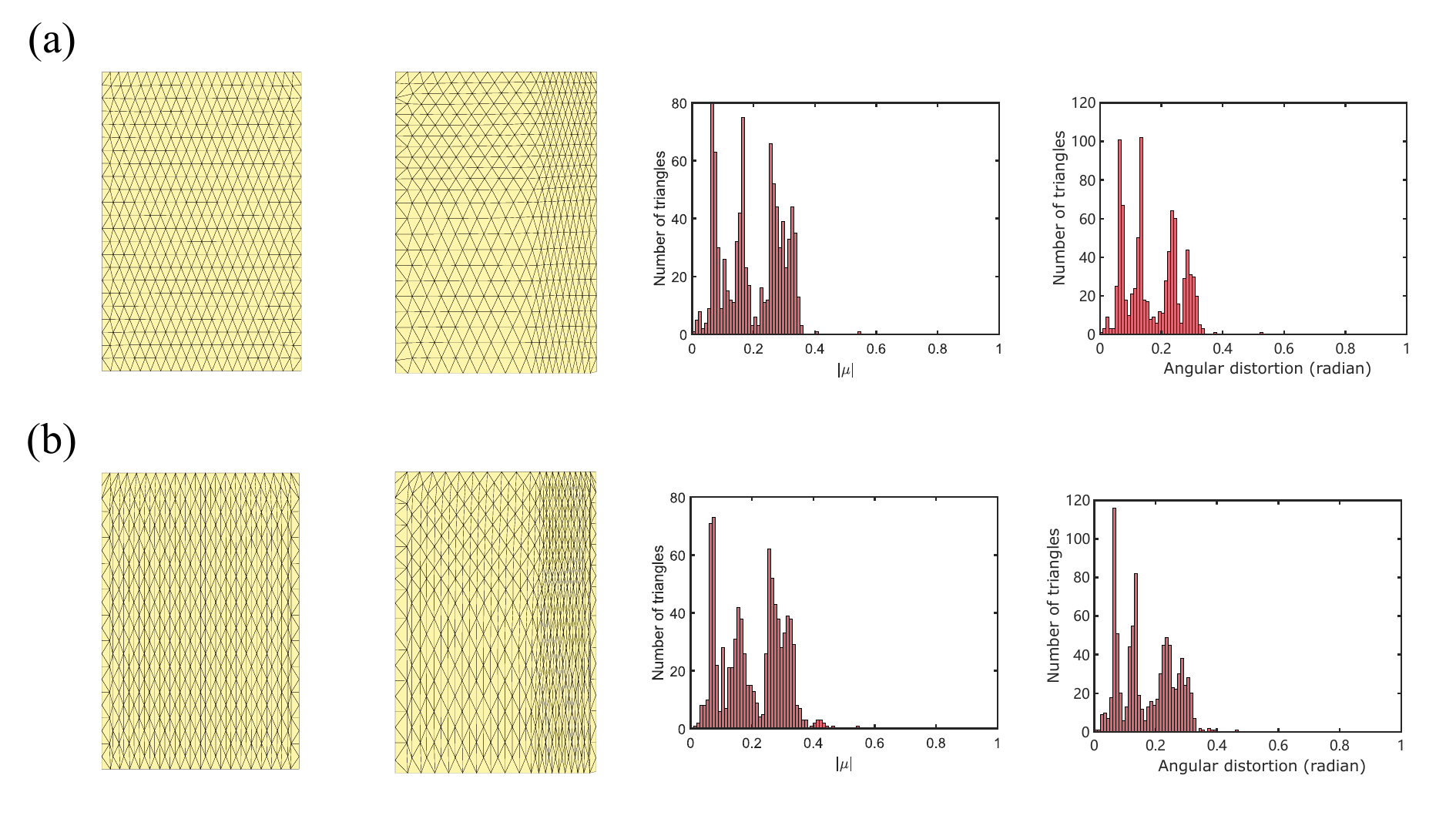}
    \caption{\textbf{Evaluation of the Beltrami coefficient and angular distortion for meshes with different quality.} Each row shows one example. (a) A Delaunay mesh with highly regular triangle elements. (b) A non-Delaunay mesh with a large number of sharp triangles. Left to right: The input mesh, the mapping result, the histogram of the norm of the Beltrami coefficient $|\mu_T|$, and the histogram of the face-based angular distortion $\epsilon_{\text{angle}_T}$.}
    \label{fig:quality}
\end{figure}

\begin{figure}[t!]
    \centering
    \includegraphics[width=\textwidth]{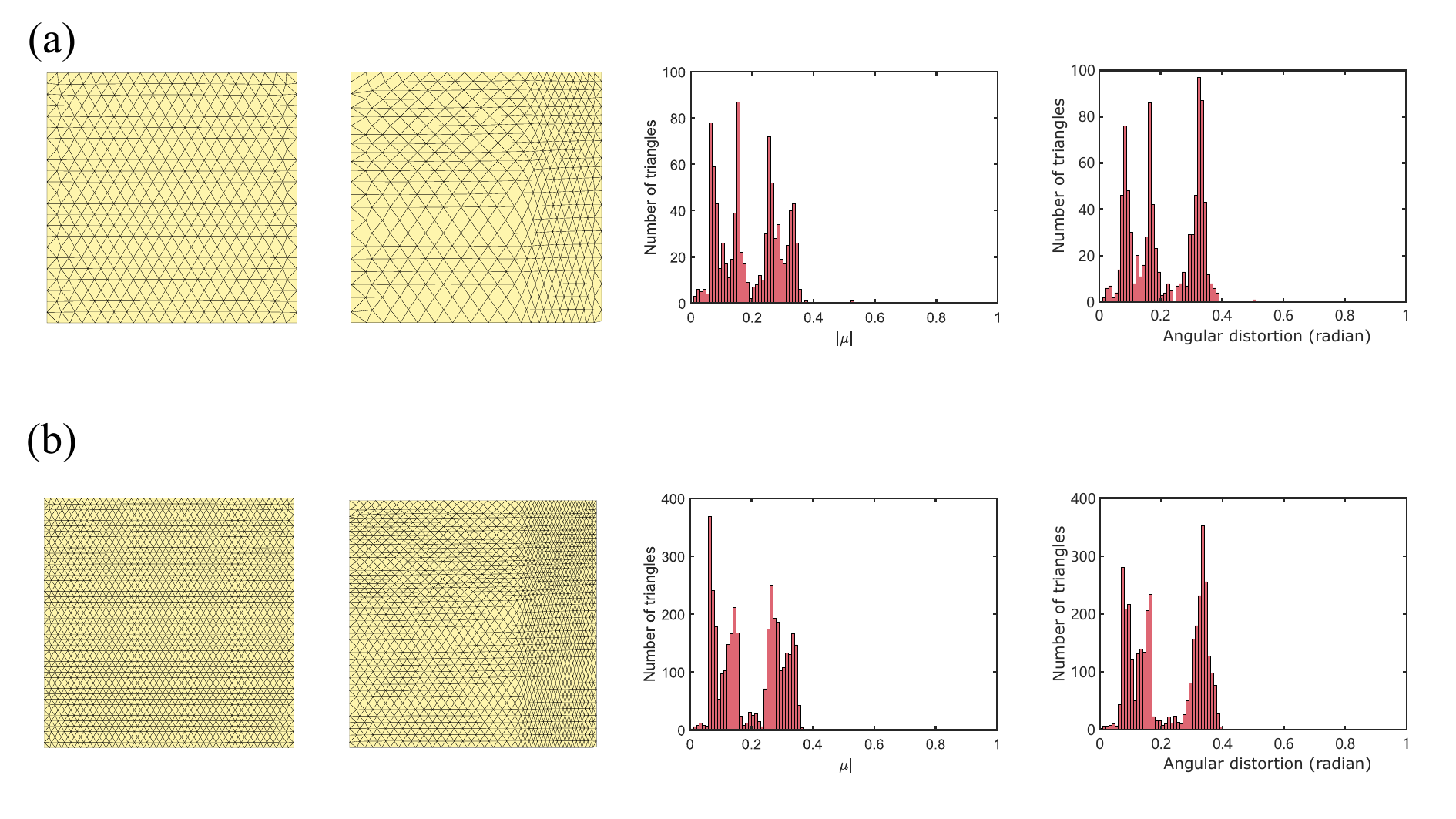}
    \caption{\textbf{Evaluation of the Beltrami coefficient and angular distortion for meshes with different resolutions.} Each row shows one example. (a) A coarse mesh with around 900 elements. (b) A denser mesh with around 3600 elements. Left to right: The input surface mesh, the mapping result, the histogram of the norm of the Beltrami coefficient $|\mu_T|$, and the histogram of the face-based angular distortion $\epsilon_{\text{angle}_T}$.}
    \label{fig:resolution}
\end{figure}

\section{Conclusion} \label{sec:conclusion}
While many geometric mapping methods have been developed over the past several decades, the focus has been primarily on the algorithmic advances. For the evaluation of the geometric distortion, many works have only considered the direct use of the angular distortion based on discrete mesh elements. In this work, we have established theoretical results to connect the Beltrami coefficient and the angular distortion in both the continuous case and the discrete case. Specifically, we have derived various estimates and bounds of the overall and maximum angular distortion in terms of the Beltrami coefficient. The estimates provide us with a simple but effective way to connect intuitive angle-based measures, which are commonly used in graphics and imaging, with more rigorous analytical tools on quasi-conformal distortions in complex analysis. Altogether, our theoretical results provide a solid foundation and a new perspective for the design and analysis of geometric mapping methods for practical applications.

Note that our current work is limited to triangle meshes. In our future work, we plan to extend our study to quadrilateral or polygonal meshes to aid the design and analysis of a wider class of mapping methods. Another future direction is to extend our study to three-dimensional (3D) volumetric mapping problems and explore the theoretical connections between the concepts and bounds of the 3D quasi-conformal distortion and the angles in the volumetric mesh elements. Finally, we plan to utilize the established results to quantify changes and deformations for shape analysis problems in biology and medicine.  

\bibliographystyle{ieeetr}
\bibliography{reference.bib}

\end{document}